\let\arxivVersion=\relax

\ifdefined\arxivVersion
  \documentclass[nonacm,format=acmsmall, review=false, screen=true]{acmart}
\else
  \documentclass[format=acmsmall, review=false, screen=true]{acmart}
\fi

\ifdefined\arxivVersion\else
  \acmJournal{TALG}
  \acmVolume{0}
  \acmNumber{0}
  \acmArticle{0}
  \acmYear{2018}
  \acmMonth{10}
  \copyrightyear{2018}

  \setcopyright{acmlicensed}

  \acmDOI{0000001.0000001}

  \received{October 2018}
\fi

\title[Cuckoo Hashing Thresholds with Overlapping Blocks]{Load Thresholds for Cuckoo Hashing with~Overlapping~Blocks}
\thanks{A preliminary version of this work appeared in the Proceedings of the 45th International Colloquium on
Automata, Languages, and Programming (ICALP), 102:1--102:10, 2018.}
\author{Stefan Walzer}
\orcid{0000-0002-6477-0106}
\affiliation{%
  \institution{Technische Universit{\"a}t Ilmenau}
  \streetaddress{Max-Planck-Ring 14}
  \city{Ilmenau}
  \state{Th{\"u}ringen}
  \postcode{98693}
  \country{Germany}}
\email{stefan.walzer@tu-ilmenau.de}



\usepackage[T1]{fontenc}
\usepackage[utf8]{inputenc}
\usepackage{amsmath,amsfonts}
\usepackage{mathtools}
\usepackage{graphicx}
\usepackage{enumerate}
\usepackage{booktabs}
\usepackage[capitalise,noabbrev]{cleveref}

\usepackage{newunicodechar}
\usepackage{silence}
\newunicodechar{⁻}{^-}
\newunicodechar{⁺}{^+}
\newunicodechar{₋}{_-}
\newunicodechar{₊}{_+}
\newunicodechar{ℓ}{\ell}
\newunicodechar{•}{\item}
\newunicodechar{…}{\dots}
\newunicodechar{≔}{\coloneqq}
\newunicodechar{≤}{\leq}
\newunicodechar{≥}{\geq}
\newunicodechar{≰}{\nleq}
\newunicodechar{≱}{\ngeq}
\newunicodechar{ϕ}{\varphi}
\newunicodechar{≠}{\neq}
\newunicodechar{¬}{\neg}
\newunicodechar{≡}{\equiv}
\newunicodechar{₀}{_0}
\newunicodechar{₁}{_1}
\newunicodechar{₂}{_2}
\newunicodechar{₃}{_3}
\newunicodechar{₄}{_4}
\newunicodechar{₅}{_5}
\newunicodechar{₆}{_6}
\newunicodechar{₇}{_7}
\newunicodechar{₈}{_8}
\newunicodechar{₉}{_9}
\newunicodechar{⁰}{^0}
\newunicodechar{¹}{^1}
\newunicodechar{²}{^2}
\newunicodechar{³}{^3}
\newunicodechar{⁴}{^4}
\newunicodechar{⁵}{^5}
\newunicodechar{⁶}{^6}
\newunicodechar{⁷}{^7}
\newunicodechar{⁸}{^8}
\newunicodechar{⁹}{^9}
\newunicodechar{∈}{\in}
\newunicodechar{∉}{\notin}
\newunicodechar{⊂}{\subset}
\newunicodechar{⊃}{\supset}
\newunicodechar{⊆}{\subseteq}
\newunicodechar{⊇}{\supseteq}
\newunicodechar{⊄}{\nsubset}
\newunicodechar{⊅}{\nsupset}
\newunicodechar{⊈}{\nsubseteq}
\newunicodechar{⊉}{\nsupseteq}
\newunicodechar{∪}{\cup}
\newunicodechar{∩}{\cap}
\newunicodechar{∀}{\forall}
\newunicodechar{∃}{\exists}
\newunicodechar{∄}{\nexists}
\newunicodechar{∨}{\vee}
\newunicodechar{∧}{\wedge}
\newunicodechar{ℝ}{\mathbb{R}}
\newunicodechar{ℕ}{\mathbb{N}}
\newunicodechar{ℤ}{\mathbb{Z}}
\newunicodechar{·}{\cdot}
\newunicodechar{∘}{\circ}
\newunicodechar{×}{\times}
\newunicodechar{↑}{\uparrow}
\newunicodechar{↓}{\downarrow}
\newunicodechar{→}{\rightarrow}
\newunicodechar{←}{\leftarrow}
\newunicodechar{⇒}{\Rightarrow}
\newunicodechar{⇐}{\Leftarrow}
\newunicodechar{↔}{\leftrightarrow}
\newunicodechar{⇔}{\Leftrightarrow}
\newunicodechar{↦}{\mapsto}
\newunicodechar{∅}{\emptyset}
\newunicodechar{∞}{\infty}
\newunicodechar{≅}{\cong}
\newunicodechar{≈}{\approx}

\newunicodechar{α}{\alpha}
\newunicodechar{β}{\beta}
\newunicodechar{γ}{\gamma}
\newunicodechar{Γ}{\Gamma}
\newunicodechar{δ}{\delta}
\newunicodechar{Δ}{\Delta}
\newunicodechar{ε}{\varepsilon}
\newunicodechar{ζ}{\zeta}
\newunicodechar{η}{\eta}
\newunicodechar{θ}{\theta}
\newunicodechar{Θ}{\Theta}
\newunicodechar{ι}{\iota}
\newunicodechar{κ}{\kappa}
\newunicodechar{λ}{\lambda}
\newunicodechar{Λ}{\Lambda}
\newunicodechar{μ}{\mu}
\newunicodechar{ν}{\nu}
\newunicodechar{ξ}{\xi}
\newunicodechar{Ξ}{\Xi}
\newunicodechar{π}{\pi}
\newunicodechar{Π}{\Pi}
\newunicodechar{ρ}{\rho}
\newunicodechar{σ}{\sigma}
\newunicodechar{Σ}{\Sigma}
\newunicodechar{τ}{\tau}
\newunicodechar{υ}{\upsilon}
\newunicodechar{ϒ}{\Upsilon}
\newunicodechar{ϕ}{\phi}
\newunicodechar{Φ}{\Phi}
\newunicodechar{χ}{\chi}
\newunicodechar{ψ}{\psi}
\newunicodechar{Ψ}{\Psi}
\newunicodechar{ω}{\omega}
\newunicodechar{Ω}{\Omega}

\usepackage{bm}
\usepackage{xspace}
\usepackage{todonotes}
\usepackage{silence}
\usepackage{tikz}
\usetikzlibrary{calc,arrows.meta}

\usepackage{tikz}
\usetikzlibrary{calc,shapes.geometric,backgrounds,positioning}

\tikzstyle{bucket}=[draw=black,fill=black!20]
\tikzstyle{window}=[inner sep=1,diamond,very thick,draw=blue,fill=white]
\tikzstyle{divider}=[circle,inner sep=1,draw=black,very thick,fill=black]
\tikzstyle{dividing line}=[draw,red!50,very thick]
\tikzstyle{object}=[draw=black,very thick,fill=white,circle,inner sep=1]
\tikzstyle{weight}=[font=\footnotesize]
\tikzstyle{separator}=[draw,thick,circle,every node/.style={draw,very thick,fill=gray!50,inner sep=1}]
\tikzstyle{vertex}=[fill,circle,inner sep=1]

\def\Idw{I_{\scalebox{0.5}{\divider} → \scalebox{0.5}{\window}}}
\def\Iwd{I_{\scalebox{0.5}{\window} → \scalebox{0.5}{\divider}}}
\def\Iow{I_{\scalebox{0.5}{\object} → \scalebox{0.5}{\window}}}
\def\Iwo{I_{\scalebox{0.5}{\window} → \scalebox{0.5}{\object}}}

\def\divider{\tikz[scale=1.5] \node[divider] {};}
\def\window{\tikz[scale=1.5] \node[window] {};}
\def\object{\tikz[scale=1.5] \node[object] {};}

\makeatletter
\g@addto@macro\@floatboxreset\centering
\makeatother

\theoremstyle{plain}
\newtheorem{proposition}{Proposition}
\newtheorem*{mainthm}{Main Theorem}
\newtheorem{conjecture}{Conjecture}
\crefname{conjecture}{Conjecture}{Conjectures}

\newcommand{\whp}{\textrm{whp}\xspace}
\newcommand{\AS}{\xspace\text{almost surely}\xspace}
\newcommand{\ASshort}{\xspace\text{a.s.}\xspace}

\DeclareMathOperator{\Po}{Po}
\DeclareMathOperator{\Bin}{Bin}
\DeclareMathOperator{\E}{\mathbb{E}}

\def\W{W_{n,cn}^{k,ℓ}}
\def\B{B_{n,cn}^{k,ℓ}}
\def\hut{\expandafter\hat}

\def\w{\mathfrak{w}}

\def\gap{\mathrm{gap}}
\def\brack#1#2{[\!\begin{smallmatrix} #1 \\ #2 \end{smallmatrix}\!]}

\begin{abstract}
We consider a natural variation of cuckoo hashing proposed by Lehman and Panigrahy (2009). Each of $cn$ objects is assigned $k = 2$ intervals of size $ℓ$ in a linear hash table of size $n$ and both starting points are chosen independently and uniformly at random. Each object must be placed into a table cell within its intervals, but each cell can only hold one object.
Experiments suggested that this scheme outperforms the variant with \emph{blocks} in which intervals are aligned at multiples of $ℓ$. In particular, the \emph{load threshold} is higher, i.e.\ the load $c$ that can be achieved with high probability. For instance, Lehman and Panigrahy (2009) empirically observed the threshold for $ℓ = 2$ to be around $96.5\%$ as compared to roughly $89.7\%$ using blocks. They pinned down the asymptotics of the thresholds for large $ℓ$, but the precise values resisted rigorous analysis.

We establish a method to determine these load thresholds for all $ℓ ≥ 2$, and, in fact, for general $k ≥ 2$. For instance, for $k = ℓ = 2$ we get $≈ 96.4995\%$.
We employ a theorem due to Leconte, Lelarge, and Massoulié (2013), which adapts methods from statistical physics to the world of hypergraph orientability. In effect, the orientability thresholds for our graph families are determined by belief propagation equations for certain graph limits. As a side note we provide experimental evidence suggesting that placements can be constructed in linear time using an adapted version of an algorithm by Khosla (2013).
\end{abstract}

\begin{document}
\maketitle

\section{Introduction}

In \textbf{standard cuckoo hashing} \cite{PR:Cuckoo:2004}, a set $X = \{x₁,…,x_{cn}\}$ of objects (possibly with associated data) from a universe $\mathcal{U}$ is to be stored in a hash table indexed by $V = \{0,…,n-1\}$ of size $n$ such that each object $x_i$ resides in one of two associated memory locations $h₁(x_i),h₂(x_i)$, given by hash functions $h₁,h₂ : \mathcal{U} → V$. In most theoretic works, these functions are 
modelled as fully random functions, selected uniformly and independently from~$V^{\mathcal{U}}$.

The load parameter $c ∈ [0,1]$ indicates the desired space efficiency, i.e.\ the ratio between objects and allocated table positions. Whether or not a valid placement of the objects in the table exists is well predicted by whether $c$ is above or below the \emph{threshold} $c^* = \frac 12$: If $c ≤ c^* - ε$ for arbitrary $ε > 0$, then a placement exists with high probability (\whp), i.e.\ with probability approaching $1$ as $n$ tends to infinity, and if $c ≥ c^* + ε$ for $ε >0$, then no placement exists \whp.

If a placement is found, we obtain a dictionary data structure representing $X ⊆ \mathcal{U}$. To check whether an object $x ∈ U$ resides in the dictionary (and possibly retrieve associated data), only the memory locations $h₁(x)$ and $h₂(x)$ need to be computed and searched for $x$. Combined with results facilitating swift creation, insertion and deletion, standard cuckoo hashing has decent performance when compared to other hashing schemes at load factors around $\frac 13$ \cite{PR:Cuckoo:2004}.

Several generalisations have been studied that allow trading rigidity of the data structure---and therefore performance of lookup operations---for load thresholds closer to $1$.

\begin{itemize}
• In \textbf{$\bm{k}$-ary cuckoo hashing}, due to Fotakis et al.~\cite{FPSS:Space_Efficient:2005}, a general number $k ≥ 2$ of hash functions is used.
•  Dietzfelbinger and Weidling \cite{DW07:Balanced:2007} propose partitioning the table into $\frac n{ℓ}$ contiguous \textbf{blocks of size $\bm{ℓ}$} 
and assign two random blocks to each object via the two hash functions, allowing an object to reside anywhere within those blocks.
• By \textbf{windows of size $\bm{ℓ}$} we mean the related idea---proposed in Lehman and Panigrahy \cite{LP:3.5-Way:2009} and the appendix in \cite{DW07:Balanced:2007}---where $x$ may reside anywhere in the intervals $[h₁(x),h₁(x)+ℓ)$ and $[h₂(x),h₂(x)+ℓ)$ (all indices understood modulo $n$). Compared to the block variant, the values $h₁(x),h₂(x) ∈ V$ need not be multiples of $ℓ$, so the possible intervals do not form a partition of $V$.
\end{itemize}
The overall performance of a cuckoo hashing scheme is a story of multidimensional trade-offs and hardware dependencies, but based on experiments in \cite{DW07:Balanced:2007,LP:3.5-Way:2009} and roughly speaking, the following empirical claims can be made:
\begin{itemize}
    • $k$-ary cuckoo hashing for $k > 2$ is slower than the other two approaches. This is because lookup operations trigger up to $k$ evaluations of hash functions and $k$ random memory accesses, each likely to result in a cache fault. In the other cases, only the number of key comparisons rises, which are comparatively cheap.
    • Windows of size $ℓ$ offer a better trade-off between worst-case lookup times and space efficiency than blocks of size $ℓ$.
\end{itemize}
Although our results are oblivious of hardware effects, they support the second empirical observation from a mathematical perspective.


\subsection{Previous Work on Thresholds}

Precise thresholds are known for
$k$-ary cuckoo hashing \cite{DGMMPR:Tight:2010,FM:Maximum:2012,FP:Orientability:2010},
cuckoo hashing with blocks of size $ℓ$ \cite{FR:The_k-orientability:2007,CSW:The_Random:2007},
and the combination of both, i.e.\ $k$-ary cuckoo hashing with blocks of size $ℓ$ with $k ≥ 3, ℓ ≥ 2$ \cite{FKP:The_Multiple:2011}. 
The techniques in the cited papers are remarkably heterogeneous and often specific to the cases at hand. Lelarge \cite{L:A_New_Approach:2012} managed to unify the above results using techniques from statistical physics that, perhaps surprisingly, feel like they grasp more directly at the core phenomena. Generalising further, Leconte, Lelarge, and Massoulié \cite{L:Belief_Propagation:2013} solved the case where each object must occupy $j ∈ ℕ$ incident table positions, $r ∈ ℕ$ of which may lie in the same block (see also \cite{GW:Load_Balancing:2010}).

Lehman and Panigrahy \cite{LP:3.5-Way:2009} showed that, asymptotically, the load threshold is $1-(2/e + o_{ℓ}(1))^{ℓ}$ for cuckoo hashing with blocks of size $ℓ$ and $1-(1/e + o_{ℓ}(1))^{1.59ℓ}$ in the case of windows, with no implication for small constant $ℓ$.  Beyer \cite{B:Analysis:2012} showed in his master's thesis that for $ℓ = 2$ the threshold is at least $0.829$ and at most $0.981$. To our knowledge, this is an exhaustive list of published work concerning windows.

In a spirit similar to cuckoo hashing with windows, Porat and Shalem \cite{PS:A_Cuckoo_Hashing:2012} analyse a scheme where memory is partitioned into pages and a bucket of size $k$ is a set of $k$ memory positions from the same page (not necessarily contiguous). The authors provide rigorous lower bounds on the corresponding thresholds as well as empirical results.




\subsection{Our Contribution}

We provide precise thresholds for $k$-ary cuckoo hashing with windows of size $ℓ$ for all $k,ℓ ≥ 2$. In particular this solves the case of $k = 2$ left open in \cite{DW07:Balanced:2007,LP:3.5-Way:2009}. Note the pronounced improvements in space efficiency when using windows over blocks, for instance in the case of $k = ℓ = 2$, where the threshold is at roughly $96.5\%$ instead of roughly $89.7\%$.

Formally, for any $k,ℓ ≥ 2$, we identify real analytic functions 
$f_{k,ℓ}$, $g_{k,ℓ}$ (see \cref{sec:messagePassing}), such that for $γ_{k,ℓ} = \inf_{λ > 0}\{f_{k,ℓ}(λ) \mid g_{k,ℓ}(λ) < 0\}$ we have

\begin{mainthm}
    The threshold for $k$-ary cuckoo hashing with windows of size $ℓ$ is $γ_{k,ℓ}$, in particular for any $ε > 0$,
    \begin{enumerate}[(i)]
        \item if $c > γ_{k,ℓ} + ε$, then no valid placement of objects exists \whp and
        \item if $c < γ_{k,ℓ} - ε$, then \,a\, valid placement of objects exists \whp.
    \end{enumerate}
\end{mainthm}
While $f_{k,ℓ}$ and $g_{k,ℓ}$ are very unwieldy, with ever more terms as $ℓ$ increases, numerical approximations of $γ_{k,ℓ}$ can be attained with mathematics software, as explained in \cref{sec:approximations}. We provide some values in \cref{table:values}.

\subsection{Methods}

The obvious methods to model cuckoo hashing with windows either give probabilistic structures with awkward dependencies or the question to answer for the structure follows awkward rules. Our first non-trivial step is to transform a preliminary representation into a hypergraph with $n$ vertices, $cn$ uniformly random hyperedges of size $k$, an added deterministic cycle, and a question strictly about the orientability of this hypergraph.

In the new form, the problem is approachable by a combination of belief propagation methods and the objective method \cite{AS:Objective_Method:2004}, adapted to the world of hypergraph orientability by Lelarge \cite{L:A_New_Approach:2012} in his insightful paper. These results were further strengthened by a Theorem in \cite{L:Belief_Propagation:2013}, which we apply at a critical point in our argument.

As the method is fundamentally about approximate sizes of incomplete orientations, it leaves open the possibility of $o(n)$ unplaced objects; a gap that can be closed in an afterthought with standard methods.

A comprehensive overview of our argument is given in \cref{sec:outline}. Key propositions are found in \cref{sec:eliminate-cycles,sec:local-weak-convergence,sec:messagePassing,sec:mainproof}.

\begin{table}[htpb]
    \def\g{\color{black!70}}
    \small
    \centering
    Thresholds $c_{k,ℓ}$ for $k$-ary cuckoo hashing with blocks of size $ℓ$:\\
    \renewcommand{\tabcolsep}{0.15cm}
    \begin{tabular}{cccccccc}
        \toprule 
        $ℓ$\textbackslash\raisebox{1.5pt}{$k$} & 2 & 3 & 4 & 5 & 6 & 7\\
        \midrule
        1 & 0.5 & 0.9179352767 & 0.9767701649 & 0.9924383913 &  0.9973795528 & 0.9990637588\\
        2 & 0.8970118682 & 0.9882014140 & 0.9982414840 & 0.9997243601 & 0.9999568737 & 0.9999933439\\
        3 & 0.9591542686 & 0.9972857393 & 0.9997951434 & 0.9999851453 & 0.9999989795 & 0.9999999329\\
        4 & 0.9803697743 & 0.9992531564 & 0.9999720661 & 0.9999990737 & 0.9999999721 & 0.9999999992\\        
        \bottomrule
    \end{tabular}\\[15pt]
    Thresholds $γ_{k,ℓ}$ for $k$-ary cuckoo hashing with windows of size $ℓ$:\\
    \begin{tabular}{cccccccc}
        \toprule 
        $ℓ$\textbackslash\raisebox{1.5pt}{$k$} & 2 & 3 & 4 & 5 & 6 & 7\\
        \midrule
        1 & 0.5 & 0.9179352767 & 0.9767701649 & 0.9924383913 &  0.9973795528 & 0.9990637588\\
        2 & 0.9649949234 & 0.9968991072 & 0.9996335076 & 0.9999529036 & 0.9999937602 & 0.9999991631\\
        3 & 0.9944227538 & 0.9998255112 & 0.9999928198 & 0.9999996722 & 0.9999999843 & 0.9999999992\\
        4 & 0.9989515932 & 0.9999896830 & 0.9999998577 & 0.9999999977 & $≈$ 1 & $≈$ 1\\
        \bottomrule
    \end{tabular}
    \caption[fragile]{Some thresholds $c_{k,ℓ}$ as obtained by \cite{PR:Cuckoo:2004,CSW:The_Random:2007,FR:The_k-orientability:2007,DGMMPR:Tight:2010,FM:Maximum:2012,FP:Sharp:2012,FKP:The_Multiple:2011}
    and values of $γ_{k,ℓ}$ as obtained from our main theorem, rounded to 10 decimal places.\\In both tables, the line for $ℓ = 1$ corresponds to plain $k$-ary cuckoo hashing, reproduced here for comparison.
    }
    \label{table:values}
    \label{table:bucketvalues}
\end{table}

\subsection{Further Discussion}

At the end of this paper, we touch on three further issues that complement our results but are somewhat detached from our main theorem.
\begin{description}
    •[Numerical approximations of the thresholds.] In \cref{sec:approximations} we explain how mathematics software can be used to get approximations for the values $γ_{k,ℓ}$, which have been characterised only implicitly.
    •[Speed of convergence.] In \cref{sec:exp-speed-of-convergence} we provide experimental results with finite table sizes to demonstrate how quickly the threshold behaviour emerges.
    •[Constructing orientations] In \cref{sec:khosla-lsa} we examine the LSA algorithm by Khosla for insertion of elements, adapted to our hashing scheme. Experiments suggest an expected constant running time per element as long as the load is bounded away from the threshold, i.e.\ $c < γ_{k,ℓ} - ε$ for some $ε > 0$.
\end{description}


\section{Definitions and Notation}
\label{sec:definitions}

A cuckoo hashing scheme specifies for each object $x ∈ X$ a set $A_x ⊂ V$ of table positions that $x$ may be placed in. For our purposes, we may identify $x$ with $A_x$. In this sense, $H = (V,X)$ is a hypergraph, where table positions are vertices and objects are hyperedges. 
The task of placing objects into admissible table positions corresponds to finding an \emph{orientation} of $H$, which assigns each edge $x ∈ X$ to an incident vertex $v ∈ x$ such that no vertex has more than one edge assigned to it. If such an orientation exists, $H$ is \emph{orientable}.

We now restate the hashing schemes from the introduction in this hypergraph framework, switching to letters $e$ (and $E$) to refer to (sets of) edges. We depart in notation, but not in substance, from definitions given previously, e.g. \cite{FPSS:Space_Efficient:2005,DW05:Balanced:2005,LP:3.5-Way:2009}. Illustrations are available in \cref{fig:k-ary-blocks-windows}.

Concerning \textbf{$\bm{k}$-ary cuckoo hashing} the hypergraph is given as:
\begin{equation}
    H_n = H_{n,cn}^k := (ℤ_n, E = \{ e₁,e₂,…,e_{cn}\}), \quad \text{for } e_i ← \brack{ℤ_n}{k},
\end{equation}
where $ℤ_n = \{0,1,…,n-1\}$ and for a set $S$ and $k ∈ ℕ$ we write $e ← \brack{S}{k}$ to indicate that $e = \{s₁,…,s_k\}$ is obtained by picking $s₁,…,s_k$ independently and uniformly at random from $S$.

There is a subtle difference to picking $e$ uniformly at random from $\binom{S}{k}$, the set of all $k$-subsets of $S$, as the elements $s₁,…,s_k$ need not be distinct. We therefore understand $e$ as a multiset. Also, we may have $e_i = e_j$ for $i ≠ j$, so $E$ is a multiset as well.\footnote{While our choice for the probability space is adequate for cuckoo hashing and convenient in the proof, such details are inconsequential. Choosing $H_n$ uniformly from the set of all hypergraphs with $cn$ \emph{distinct} edges all of which contain $k$ \emph{distinct} vertices would be equivalent for our purposes.}

\begin{figure}[htbp]
    \includegraphics[page=1]{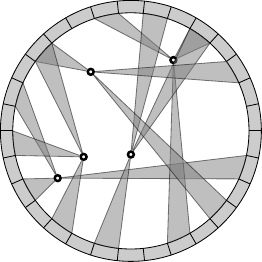}\ 
    \includegraphics[page=2]{img/indep-blocks-windows.pdf}\ 
    \includegraphics[page=3]{img/indep-blocks-windows.pdf}
    \Description{Three rings subdivided into 30 segments each. In the second ring, the segments come in 15 blocks of 2 segments each. There are five dots within each ring. Each dot is connected to three segments / three blocks of two segments / three pairs of neighbouring segments, respectively.}
    \caption{Drawing of possible outcomes for the hypergraphs $H_n$, $B_n$ and $W_n$ (modelling $k$-ary cuckoo hashing plain / with\,blocks / with\,windows) for $n = 30$, $c = \frac 16$, $k = 3$ and $ℓ = 2$ ($ℓ$ only for $B$ and $W$). Each edge is drawn as a point and connected to all incident table cells, which are arranged in a circle. In the case of $B$, thick lines indicate the borders between blocks.}
    \label{fig:k-ary-blocks-windows}
\end{figure}

Assuming the table size $n$ is a multiple of $ℓ$, \textbf{$\bm{k}$-ary cuckoo hashing with blocks of size $\bm{ℓ}$} is modelled by the hypergraph

\begin{equation}
    B_n = \B := (ℤ_n, \{e'₁,e'₂,…,e'_{cn}\}), \ \  \text{where } e'_i = \bigcup_{j ∈ e_i} [jℓ,(j+1)ℓ) \text{ and } e_i ← \brack{ℤ_{n/ℓ}}{k},
    \label{eq:blockgraph}
\end{equation}
that is, each hyperedge is the union of $k$ blocks chosen uniformly at random from the set of all blocks, which are the $n/ℓ$ intervals of size $ℓ$ in $ℤ_n$ that start at a multiple of $ℓ$.
Note that for $ℓ = 1$ we recover $H_n$. 

Similarly, \textbf{$\bm{k}$-ary cuckoo hashing with windows of size $\bm{ℓ}$} is modelled by
\begin{equation}
    W_n = \W := (ℤ_n, \{e'₁,e'₂,…,e'_{cn}\}), \ \  \text{where } e'_i = \bigcup_{j ∈ e_i} [j,j+ℓ) \text{ and } e_i ← \brack{ℤ_n}{k},
    \label{eq:defW}
\end{equation}
that is, each hyperedge is the union of $k$ windows chosen uniformly at random from the set of all windows, which are the $n$ intervals of size $ℓ$ in $ℤ_n$, this time without alignment restriction. Note that intervals wrap around at the ends of the set $\{0,…,n-1\}$ with no awkward “border intervals”.
Again, for $ℓ = 1$ we recover $H_n$.


\section{Outline of the Proof}
\label{sec:outline}

\def\substep#1{\paragraph{\textbf{#1}}}
\def\dir{\mathrel{\tikz[baseline] \draw[->,>=Latex] (0,3pt) -- ++(12pt,0);}}
\def\hdir{\mathrel{\tikz[baseline] \draw[->,>=Latex] (0,3pt) -- node[pos=0.3,below]{\smash{$\hat{}$}}++(12pt,0);}}

\substep{Step 1: A tidier problem.}

The elements of an edge $e$ of $B_n$ and $W_n$ are not independent, as $e$ is the union of $k$ intervals of size $ℓ$. This poorly reflects the actual tidiness of the probabilistic object. We may obtain a model with independent elements in edges, by switching to a more general notion of what it means to orient a hypergraph.


Formally, given a weighted hypergraph $H = (V,E,η)$ with weight function $η : V ∪ E → ℕ$, an \emph{orientation} $μ$ of $H$ assigns to each pair $(e,v)$ of an edge and an incident vertex a number $μ(e,v) ∈ ℕ₀$ such that
\begin{equation}
    ∀e ∈ E\colon \sum_{v ∈ e} μ(e,v) = η(e), \text{ and } ∀v ∈ V\colon \sum_{e \ni v} μ(e,v) ≤ η(v).\label{eq:defOrientation}
\end{equation}
We will still say that an edge $e$ \emph{is oriented to} a vertex $v$ (possibly several times) if $μ(e,v) > 0$. One may be inclined to call $η(v)$ a \emph{capacity}  for $v ∈ V$ and $η(e)$ a \emph{demand} for $e ∈ E$, but we use the same letter in both cases as the distinction is dropped later anyway.

Orientability of $H,B$ and $W$ from earlier is also captured in the generalised notion with implicit weights of $η ≡ 1$.

\begin{figure}[htbp]
    \def\whoop#1{\raisebox{1cm}{#1}}
    \begin{tabular}{c@{\hspace{-13pt}}c@{}c@{}c}
        \whoop{\textbf{(a)}}
        &\includegraphics[page=1,scale=.9]{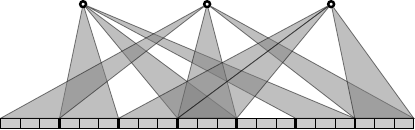}
        &\whoop{$→$}
        &\includegraphics[page=2,scale=.9]{img/block-simplification.pdf}
        \\
        \whoop{\textbf{(b)}}
        &\includegraphics[page=1,scale=.9]{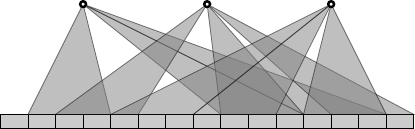}
        &\whoop{$→$}
        &\includegraphics[page=5,scale=.9]{img/make-windows-explicit.pdf}
    \end{tabular}
    \Description{(a) illustrates the transformation from $B_n$ to $\hut B_n$, (b) the transformation from $W_n$ to $\hut W_n$.}
    \caption[fragile]{\textbf{(a)} In $k$-ary cuckoo hashing with blocks of size $ℓ$ (here $k = ℓ = 3$), we can contract each block into a single vertex of weight $ℓ$ to obtain a simpler but equivalent representation of the orientation problem.\\
    \textbf{(b)} In $k$-ary cuckoo hashing with windows of size $ℓ$, a similar idea works, but additional helper edges (drawn as \scalebox{1.5}{\divider}) of weight $ℓ-1$ are needed (see \cref{prop:eliminate-cycles}).\\
    }
    \label{fig:donning-the-hat}
\end{figure}


A simplified representation of $B_n$ is straightforward to obtain. We provide it mainly for illustration purposes, see \cref{fig:donning-the-hat}(a):
\begin{gather}
    \hut B_n := \hut\B := (ℤ_{n/ℓ}, \{e₁,e₂,…,e_{cn}\}, η), \quad \text{where }  e_i ← \brack{ℤ_{n/ℓ}}{k}\\
     \text{ vertex weight $η(v) = ℓ$ for $v ∈ ℤ_{n/ℓ}$ and edge weight $η(e_i) = 1$ for $1 ≤ i ≤ cn$}.\notag
\end{gather}
In $\hut B_n$, each group of $ℓ$ vertices of $B_n$ representing one block is now contracted into a single vertex of weight $ℓ$ and edges contain $k$ independent vertices representing blocks instead of $kℓ$ dependent vertices. It is clear that $B_n$ is orientable if and only if $\hut B_n$ is orientable.


In a similar spirit we identify a transformed version $\hut W_n$ for $W_n$, but this time the details are more complicated as the vertices have an intrinsic linear geometry, whereas $B_n$ featured essentially an unordered collection of internally unordered blocks. The \emph{ordinary edges} in $\hut W_n$ also have size $k$ instead of size $kℓ$, but we need to introduce additional \emph{helper edges} that capture the linear geometry of $ℤ_n$, see \cref{fig:donning-the-hat}(b). We define:
\begin{gather}
    \hut W_n := \hut\W := (ℤ_n,  C_n ∪ \{e₁,…,e_{cn}\},η)\label{eq:defhatW}\\
    \text{ with ordinary edges } e_i ← \brack{ℤ_n}{k}, \text{ helper edges }  C_n = \{h_i := (i,i+1) \mid i ∈ ℤ_n\},\notag\\
    \text{ vertex weight } η(v) = ℓ\ \text { for }\ v ∈ ℤ_n,\notag\\
    \text{ and edge weights } η(h) = ℓ-1,\ η(e) = 1\ \text { for }\ h ∈ C_n, e ∈ \{e₁,…,e_{cn}\}. \notag
\end{gather}
Note that formally the graphs $W_n$ and $\hut W_n$ are random variables on a common probability space. An outcome $ω = (e_i)_{1 ≤ i ≤ cn}$ from this space determines both graphs.

The following proposition justifies the definition. It is proved in \cref{sec:eliminate-cycles}.
\begin{proposition}
    \label{prop:eliminate-cycles}
    $\hut W_n$ is orientable if and only if  $W_n$ is orientable.\footnote{Formally this should read: The events $\{ W_n \text{ is orientable}\}$ and $\{\hut W_n \text{ is orientable}\}$ coincide.}
\end{proposition}
An important merit of $\hut W_n$ that will be useful in Step 3 is that it is \emph{locally tree-like}, meaning each vertex has a probability of $o(1)$ to be involved in a constant-length cycle. Here, by a cycle in a hypergraph we mean a sequence of distinct edges $e₁,e₂,…,e_j$ such that two consecutive edges share a vertex and $e_j$ and $e₁$ share a vertex.

Note the interesting special case $\hut W_{n,cn}^{2,2}$, which is a cycle of length $n$ with $cn$ random chords, unit edge weights and vertices of weight $2$. Understanding the orientability thresholds for this graph seems interesting in its own right, not just as a means to understand $W_{n,cn}^{2,2}$.

\substep{Step 2: Incidence Graph and Allocations.}
The next step is by no means a difficult or creative one, we merely perform the necessary preparations needed to apply \cite{L:Belief_Propagation:2013}, introducing their concept of an allocation in the process.

This will effectively get rid of the asymmetry between the roles of vertices and edges in the problem of orienting $\hut W_n$, by switching perspective in two simple ways.
The first is to consider the incidence graph $G_n$ of $\hut W_n$ instead of $\hut W_n$ itself, i.e.\ the weighted bipartite graph
\begin{equation}
    G_n = G_{n,cn}^{k,ℓ} = (\underbrace{C_n}_{A_C} ∪ \underbrace{\{e₁,…,e_{cn}\}}_{A_R}, \underbrace{ℤ_n}_B, \underbrace{\text{“$\ni$”}}_{E(G_n)}, \ \ η \ \ ).\label{def:G_n}
\end{equation}
We use $A = A_C ∪ A_R$ to denote those vertices of $G_n$ that were edges in $\hut W_n$, and $B$ for those vertices of $G_n$ that were vertices in $\hut W_n$. Vertices $a ∈ A$ and $b ∈ B$ are adjacent in $G_n$ if $b ∈ a$ in $\hat W_n$. The weights $η$ on vertices and edges in $\hut W_n$ are now vertex weights with $η(a_C) = ℓ-1$, $η(a_R) = 1$, $η(b) = ℓ$ for $a_C ∈ A_C$, $a_R ∈ A_R$, $b ∈ B$. 
The notion of $μ$ being an orientation translates to $μ$ being a map $μ\colon E(G_n) → ℕ₀$ such that $\sum_{b ∈ N(a)} μ(a,b) = η(a)$ for all $a ∈ A$ and $\sum_{a ∈ N(b)} μ(a,b) ≤ η(b)$ for all $b ∈ B$. Note that vertices from $A$ need to be \emph{saturated} (“$= η(a)$” for $a ∈ A$) while vertices from $B$ need not be (“$≤ η(b)$” for $b ∈ B$). This leads to the second switch in perspective.

Dropping the saturation requirement for $A$, we say $μ$ is an \emph{allocation} if $\sum_{u ∈ N(v)} μ(u,v) ≤ η(v)$ for all $v ∈ A ∪ B$.\label{def:allocation}




Clearly, any orientation is an allocation, but not vice versa; for instance, the trivial map $μ ≡ 0$ is an allocation. Let $|μ|$ denote the size of an allocation, i.e.\ $|μ| = \sum_{e ∈ E} μ(e)$. By bipartiteness, no allocation can have a size larger than the total weight of $A$, i.e.
\[ \text{for all allocations $μ$}\colon |μ| ≤ η(A) = \sum_{a ∈ A} η(a) = |A_C|·(ℓ-1)+|A_R|·1 = (ℓ-1+c)n\]
and the orientations of $G_n$ are precisely the allocations of size $η(A)$. We conclude:

\begin{proposition}
    \label{prop:allocations-vs-orientations}
    Let $M(G_n)$ denote the maximal size of an allocation of $G_n$. Then
    \[ \tfrac{M(G_n)}{n} = ℓ-1+c\quad \text{if and only if} \quad \text{$\hut{W}_n$ is orientable.}\] 
\end{proposition}

\substep{Step 3: The Limit $\bm{T}$ of $\bm{G_n}$.} Reaping the benefits of step $1$, we find $G_n$ to have $O(1)$ cycles of length $O(1)$ \whp. To capture the local appearance of $G_n$ even more precisely, let the \emph{$r$-ball} around a vertex $v$ in a graph be the subgraph induced by the vertices of distance at most $r$ from $v$. Then the $r$-ball around a random vertex of $G_n$ is distributed, as $n$ gets large, more and more like the $r$-ball around the root of a random infinite rooted tree $T = T_c^{k,ℓ}$. The tree $T$ is distributed as follows, with weighted nodes of types $A_C,A_R$ or $B$.\label{def:T}



\begin{itemize}
    • The root of $T$ is of type $A_C$, $A_R$ or $B$ with probability $\frac{1}{2+c}$, $\frac{c}{2+c}$ and $\frac{1}{2+c}$, respectively.
    • If the root is of type $A_C$, it has two children of type $B$. If it is of type $A_R$, it has $k$ children of type $B$. If it is of type $B$, it has two children of type $A_C$ and a random number $X$ of children of type $A_R$, where $X \sim \Po(kc)$. Here $\Po(λ)$ denotes the Poisson distribution with parameter $λ$.
    • A vertex of type $A_C$ that is not the root has one child of type $B$. A vertex of type $A_R$ that is not the root has $k-1$ children of type $B$.
    • A vertex of type $B$ that is not the root has a random number $X$ of children of type $A_R$, where $X \sim \Po(kc)$. If its parent is of type $A_C$, then it has one child of type $A_C$, otherwise it has two children of type $A_C$.
    • 
    Vertices of type $A_C$, $A_R$ and $B$ have weight $ℓ{-}1$, $1$ and $ℓ$, respectively.
\end{itemize}
All random decisions should be understood to be independent. A type is also treated as a set containing all vertices of that type. In \cref{sec:local-weak-convergence} we briefly recall the notion of \emph{local weak convergence} and argue that the following holds:

\begin{proposition}
    \label{prop:local-weak-convergence}
    Almost surely\footnote{Meaning: With probability $1$.}, $(G_n)_{n∈ℕ} = (G_{n,cn}^{k,ℓ})_{n∈ℕ}$ converges locally weakly to $T = T_c^{k,ℓ}$.
\end{proposition}

\substep{Step 4: The Method of Lelarge.} We are now in a position to apply a powerful theorem due to Leconte, Lelarge, and Massoulié \cite{L:Belief_Propagation:2013} that characterises $\lim_{n→∞} \frac{M(G_n)}{n}$ in terms of solutions to belief propagation equations for $T$. Put abstractly: The numerical limit of a function of $G_n$ is expressed as a function of the graph limit of $G_n$. We elaborate on details and deal with the equations in \cref{sec:messagePassing}. After condensing the results into a characterisation of $γ_{k,ℓ} ∈ (0,1)$ in terms of “well-behaved” functions we obtain:
\begin{proposition}\ \vspace{-\baselineskip}
    \label{prop:threshold-up-to-on}
    \[ \lim_{n → ∞}\tfrac{M(G_{n,cn})}{n}\ \begin{cases}
        \ = ℓ - 1 + c\quad \AS &\quad \text{if } c < γ_{k,ℓ}\\
        \ < ℓ - 1 + c\quad \AS &\quad \text{if } c > γ_{k,ℓ}.
    \end{cases}\]
\end{proposition}

\substep{Step 5: Closing the Gap.} It is important to note that we are not done, as
\begin{equation}
    \lim_{n→∞}\tfrac{M(G_{n,cn})}{n} = ℓ-1+c\ \ \ASshort \text{\quad does not imply\quad } M(G_{n,cn}) = n·(ℓ-1+c)\ \whp.\label{eq:peskydistinction}
\end{equation}
We still have to exclude the possibility of a gap of size $o(n)$ on the right hand side; imagine for instance $M(G_{n,cn}) = (ℓ-1+c)n - \sqrt{n}$ to appreciate the difference.
In the setting of cuckoo hashing with double hashing \cite{Leconte:Cuckoo:2013}, the analogue of this pesky distinction was in the way of obtaining precise thresholds, until recently tamed by a somewhat lengthy case analysis in \cite{MPW:CuckooDoubleHashing:2018}. We should therefore treat this carefully.


Luckily the line of reasoning by Lelarge \cite{L:A_New_Approach:2012} can be adapted to our more general setting. The key is to prove that if not all objects can be placed into the hash table,
then the configuration causing this problem has size $Θ(n)$ (and those large overfull structures do not go unnoticed on the left side of (\ref{eq:peskydistinction})).
\begin{lemma}
    \label{lem:no-small-hall-wittnesses}
    There is a constant $δ > 0$ such that \whp no set of\ \ $0 < t < δn$ vertices in $\hat{W}_n$ (of weight $ℓt$) induces edges of total weight $ℓt$ or more, provided $c ≤ 1$.
\end{lemma}
The proof of this Lemma (using first moment methods) and the final steps towards our main theorem are found in \cref{sec:mainproof}. 

\ifdefined\extendedAbstract
\else
    The following sections provide the details of the technical argument. Conclusion, outlook and acknowledgements can be found at the end of the paper.
\fi


\section{\texorpdfstring{Equivalence of $\bm{W_n}$ and $\bm{\hut W_n}$ with respect to orientability}{Equivalence of W and W hat with respect to orientability}}
\label{sec:eliminate-cycles}

\begin{figure}[htbp]
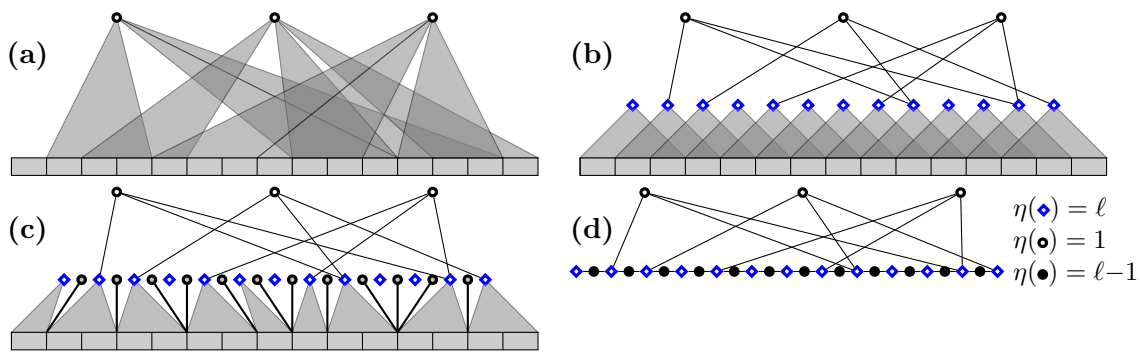

    \def\letter#1{\raisebox{1.5cm}{\textbf{(#1)}}}
    \begin{tabular}{c@{\hspace{-0.5cm}}cc@{\hspace{-1cm}}c}
        \letter{a}
        &\includegraphics[page=1,scale=.9]{img/make-windows-explicit.pdf}
        &\letter{b}
        &\includegraphics[page=2,scale=.9]{img/make-windows-explicit.pdf}\\
        \letter{c}
        &\includegraphics[page=3,scale=.9]{img/make-windows-explicit.pdf}
        &\letter{d}
        &\includegraphics[page=4,scale=.9]{img/make-windows-explicit.pdf}
    \end{tabular}

    \Description{In (a) there is a linear arangement of 15 boxes. Three circular marks are connected to three intervals of $3$ boxes each. In (b) there is a diamond mark above each interval of $3$ boxes. These are connected to those three boxes. The circular marks are now connected to the corresponding diamond marks. In (c) there are additional gray marks in between the diamond marks, each pointing to a position in between two boxes. The diamond marks are now only connected to the boxes in between the positions indicated by the two neighbouring gray marks. In (d) the boxes are gone and the gray marks are connected to the two neighbouring diamond marks.}
    \caption{Drawing of (a part of) $W_n$ in \textbf{(a)} and $\hut W_n$ in \textbf{(d)} with two intermediate ideas implicit in the proof of \cref{prop:eliminate-cycles}.}
    \label{fig:make-acyclic}
\end{figure}

In this section we prove \cref{prop:eliminate-cycles}, i.e.\ we show that $W_n$ is orientable if and only if $\hut W_n$ is orientable. Recall the relevant definitions in \cref{eq:defW,eq:defhatW,eq:defOrientation}.
Some ideas are illustrated in \cref{fig:make-acyclic}.
We start with the hypergraph $W_n$ in \textbf{(a)}. From this, \textbf{(b)} is obtained by introducing one “broker-vertex” (\scalebox{1.5}{\window}) for each interval of size $ℓ$ in the table, through which the incidences of the objects (\scalebox{1.5}{\object}) are “routed” as shown. The purpose of each broker-vertex is to “claim” part of its interval on behalf of incident objects.
To manage these claims, we imagine a “separator” (\scalebox{1.5}{\divider}) between each pair of adjacent broker-vertices that, by pointing between two table cells, indicates where the claim of one broker-vertex ends and the claim of the next broker-vertex begins, see \textbf{(c)}. There are $ℓ$ possible “settings” for each separator. The separators can be modelled as edges of weight $ℓ-1$ with $ℓ$ possible ways to distribute this weight among the two incident broker-vertices that have weight $ℓ$. The table is then fully implicit, which gives $\hut W_n$ in \textbf{(d)}.


\begin{proof}[Proof of \cref{prop:eliminate-cycles}] We introduce the shorthand $\w_i := [i,i+ℓ)$ for this proof.
    \begin{description}
        •[$⇒$] Let $μ$ be an orientation of $W_n$. We will define an orientation $\hat{μ}$ of $\hut W_n$. Recall from \cref{eq:defW,eq:defhatW} how an edge $e' = \bigcup_{j ∈ e} \w_j$ of $W_n$ is defined in terms of an edge $e$ of $\hut W_n$. If $μ$ directs $e'$ to a table cell $x ∈ ℤ_n$, we pick $j ∈ e$ with $x ∈ \w_j$. We let $\hat{μ}$ direct $e$ to $j$ and also assign to $x$ the \emph{label} $\w_j$.
        
         Note that, since $μ$ is an orientation, each $x ∈ ℤ_n$ receives at most one label this way, and that the label stems from $\{\w_{x-ℓ+1},…,\w_x\}$.
        
        We still have to orient the helper edges $h_i = (i,i+1)$ of weight $ℓ{-}1$. For this, we count the number $r_i$ of elements in $\w_i ∩ \w_{i+1}$ with a label that is to the \emph{right}, i.e.\ stems from $\{\w_{i+1},\w_{i+2},…,\w_{i+ℓ-1}\}$. We then set $\hat{μ}(h_i, i) = r_i$ and $\hat{μ}(h_i, i+1) = ℓ{-}1-r_i$.
        
        We now check that the weight of any vertex $i ∈ ℤ_n$ is respected, i.e.\ we check that
        \[ \hut{μ}(h_i,i) + \hut{μ}(h_{i-1},i) + \sum_{\mathclap{\substack{\text{ ordinary edge $e$}\\i ∈ e}}} \hut{μ}(e,i) \stackrel{!}{≤} ℓ.\]
        From ordinary edges, the contribution is $1$ for each $x ∈ \w_i$ with label $\w_i$. From $h_i$ the contribution is the number of $x ∈ \w_i ∩ \w_{i+1}$ with label in $\{\w_{i+1}, …, \w_{i+ℓ-1}\}$ and from $h_{i-1}$, the contribution is the number of $x ∈ \w_{i-1} ∩ \w_i$ \emph{not} having a label from $\{\w_i, …, \w_{i+ℓ-2}\}$. The three conditions are clearly mutually exclusive, so each $x ∈ \w_i$ can contribute at most $1$, giving a total contribution of at most $|\w_i| = ℓ$ as required.

        •[$⇐$] Let $\hat{μ}$ be an orientation of $\hut W_n$. Define $s_i := \hat{μ}(h_{i-1}, i)$ and $t_i := \hat{μ}(h_{i}, i)$. Let further $\bar{\w}_i := [i+s_i,i+ℓ-t_i) ⊆ \w_i$.
        
        Crucially, $\{\bar{\w}_i \mid i ∈ ℤ_n\}$ forms a partition of $ℤ_n$. This follows from the following properties:
        \begin{gather*}
                \max(\bar{\w}_i) = i + ℓ - t_i -1, \quad \min(\bar{\w}_{i+1}) = i+1+s_{i+1}\\
                \text{and} \quad t_i + s_{i+1} = \hat{μ}(h_i, i) + \hat{μ}(h_i, i+1) = η(h_i) = ℓ-1.
        \end{gather*}
        Here, if cyclic intervals span the “seam” of the cycle, $\max$ and $\min$ should be reinterpreted in the natural way. Now let $e^{(i)}₁,…,e^{(i)}_{ρ_i}$ be the ordinary edges directed to $i$ by $\hat{μ}$. Since $\hat{μ}$ respects $η(i) = ℓ$ we have $ρ_i + s_i + t_i ≤ ℓ$ , so $ρ_i ≤ ℓ - s_i - t_i = |\bar{\w}_i|$. We can now define the orientation $μ$ of $W_n$ to direct each $e^{(i)}_j{}'$ to $i+s_i+j - 1∈ ℤ_n$ for $1 ≤ j ≤ ρ_i$ and $i ∈ ℤ_n$.\qedhere
    \end{description}
\end{proof}


\section{Local weak convergence of \texorpdfstring{$\bm{G_n}$ to $\bm{T}$}{Gn to T}}
\label{sec:local-weak-convergence}

Recall the definitions of the finite graph $G_n$ in \cref{def:G_n} and the infinite rooted tree $T$ in Step 3 of \cref{def:T}. We obtain the rooted graph $G_n(∘)$ from $G_n$ by distinguishing one vertex---the root---uniformly at random.
For any rooted graph $R$ and $d ∈ ℕ$, let $R|_d$ denote the rooted subgraph of $R$ induced by the vertices at distance at most $d$ from the root. We treat two rooted graphs as equal if there is an isomorphism between them preserving root and vertex types. Refer to \cref{fig:G-and-T-locally} for a possible outcome of $T|₃$ and $G_n(∘)|₃$.



    
    With this notation, we can clarify \cref{prop:local-weak-convergence}, i.e.\ what we mean by saying that almost surely $(G_n)_{n ∈ ℕ}$ \emph{converges locally weakly to} $T$, namely
    \begin{equation}
        ∀d ∈ ℕ\colon ∀\text{ rooted graph $H$}\colon \lim_{n→∞}\Pr_{∘ ∈ V(G_n)}[G_n(∘)|_d = H] = \Pr[T|_d = H] \text{ almost surely}.\label{eq:local-weak-convergence}
    \end{equation}
    Here “almost surely” refers to the randomness in choosing the sequence $(G_n)_{n∈ ℕ}$ while $\Pr_{∘ ∈ V(G_n)}$ refers to the randomness in rooting $G_n$.
%
For a more general definition and a plethora of illustrating examples, refer to the excellent survey on the \emph{objective method} by Aldous and Steele \cite{AS:Objective_Method:2004}. 



The proof of (\ref{eq:local-weak-convergence}) is technical but standard. We will take the time to sketch the proof of the slightly simpler statement:
\begin{equation}
        ∀d ∈ ℕ\colon ∀\text{ rooted graph $H$}\colon \lim_{n→∞}\Pr_{G_n, ∘ ∈ V(G_n)}[G_n(∘)|_d = H] = \Pr[T|_d = H].\tag{\ref{eq:local-weak-convergence}'}\label{eq:local-weak-convergence-prime}
\end{equation}
Note that in (\ref{eq:local-weak-convergence}) the probability in the limit is a random variable (depending on $(G_n)_{n∈ℕ}$), while the probability in (\ref{eq:local-weak-convergence-prime}) is not.
 We refer to \cite{Leconte:Cuckoo:2013} where a similar proof is done in full, including concentration arguments to settle the difference between the analogues of (\ref{eq:local-weak-convergence}) and (\ref{eq:local-weak-convergence-prime}).

 \begin{figure}[htbp]
    \includegraphics{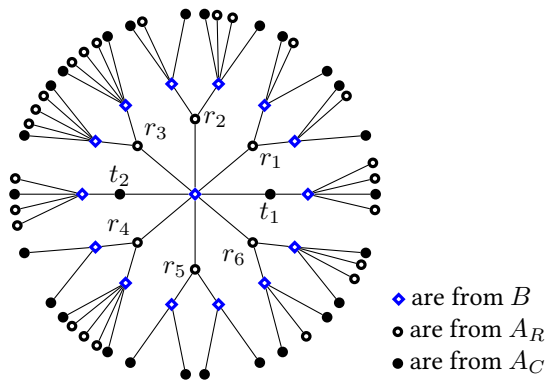}
    \Description{A tree drawn with a (root) vertex from $B$ in the center, growing in all directions. Its children are drawn around it, there are two vertices $t₁$ and $t₂$ from $A_C$ as well as six vertices $r₁,…,r₆$ from $A_R$. The vertices $t₁$ and $t₂$ have one child each, the vertices $r₁,…,r₆$ have two children each, all of which are from $B$. All these vertices have a random number of children from $A_R$ and exactly two neighbours from $A_C$ (i.e. $1$ or two children from $A_C$). All vertices at distance $3$ from the root have no further children.}
    \caption{One possibility of what $T|₃$ may look like for $k = 3$. Since the infinite random tree $T$ is designed to reflect the local characteristics of $G_n$, it is also a possibility for $G_n(∘)|₃$. Actually, the distributions of $G_n(∘)|₃$ and $T|₃$ are asymptotically equal.
    }
    \label{fig:G-and-T-locally}
 \end{figure}
 
 \begin{proof}[Proof of \cref{eq:local-weak-convergence-prime}]
    For simplicity we shall assume that child vertices of the same type are ordered (for $G_n(∘)$ and $T$ just fix an ordering at random). Correspondingly, we prove the more fine-grained version of (\ref{eq:local-weak-convergence-prime}) where equality indicates the existence of an isomorphism preserving roots, types and child orderings.
    
    Let $H$ be a possible outcome of $T|_d$, let $t ∈ \{A_C, A_R, B\}$ be the type of the root of $H$ and let $v₁,…,v_s$ be the vertices of type $B$ in $H$, except for leaves, in breadth-first-search ordering. Let further $x_i ∈ ℕ₀$ denote the number of $v_i$'s children of type $A_R$ for $1 ≤ i ≤ s$. Checking the definition of $T$, the sequence $(t,x₁,…,x_s)$ contains all random decisions that $T$ has to “get right” in order for $T|_d$ to coincide with $H$, i.e.
    \[ \Pr[T|_d = H] = \Pr[\mathrm{root}(T) ∈ t] · \prod_{i = 1}^s \Pr[X_i = x_i]\]
    where $X_i \sim \Po(ck)$ for $i ∈ ℕ$ are the independent random variables used in the construction of $T$, in breadth-first-search order.
    
    To compare this to $\Pr_{G_n, ∘ ∈ G_n}[G_n(∘)|_d = H]$, we shall reveal the type of $∘$ and then the neighbourhoods of vertices from $G_n(∘)|_d$ one by one in breadth-first-search order and check equality with $H$. Let $w₁,w₂,…$ be the vertices of type $B$ in breadth-first-search ordering and $C(w_i) ≔ N(w_i) \setminus \mathrm{parent}(w_i)$ the neighbourhood of $w_i$, except for the vertex from which $w_i$ was discovered (if $w_i ≠ {∘}$). Recall that the $cn$ vertices from $A_R$ each choose $k$ among the $n$ vertices of $B$ uniformly at random, so for $v ∈ B$ we have $|N(v) ∩ A_R| \sim \Bin(ckn, \frac 1n)$. However, when revealing $Y_i := |C(w_i) ∩ A_R|$, then a constant number $α_i$ of incidences of vertices from $A_R$ are already revealed and the full neighbourhoods of a constant number $β_i$ of vertices of type $B$ are already revealed. Thus, conditioned on $G_n(∘)$ matching with $H$ until before $N(w_i)$ is revealed, we have $Y_i \sim \Bin(ckn - α_i, \frac{1}{n-β_i})$. A second complication is that $G_n(∘)$ can contain cycles. Therefore, whenever we reveal the identity of a vertex of type $A_R$ we shall assume that it is not one of the vertices already seen and when we reveal the identity of a vertex of type $B$ found as a child of a vertex of type $A_R$, we shall assume that its position in $ℤ_n$ is not within distance $d$ of a vertex of type $B$ already seen. The probability of these events is clearly $1 - o(1)$ since only a constant number of vertices are forbidden. It is important to note that, since $H$ is a possible outcome for $T|_d$, all remaining aspects of $H$ and $G_n(∘)$ coincide by construction (e.g. the degree of vertices of types $A_R$ and $A_C$). We get
    \[ \Pr_{G_n, ∘ ∈ V(G_n)}[G_n(∘)|_d = H] = (1-o(1))\Pr[{∘} ∈ t] · \prod_{i = 1}^s \Pr[Y_i = x_i \mid \text{$G_n(∘)$ matches with $H$ up to $v_i$}].\]
    Since $\Pr[∘ ∈ t] = \Pr[\mathrm{root}(T) ∈ t]$ by construction and due to the convergence of Binomial to Poisson random variables, we get $\lim\limits_{n→∞}\Pr_{G_n, ∘}[G_n(∘)|_d = H] = \Pr[T|_d = H]$ as desired.
 \end{proof}

\section{Belief Propagation on the Limiting Tree \texorpdfstring{$\bm{T}$}{T}}
\label{sec:messagePassing}

\def\Idw{I_{A_C → B}}
\def\Iwd{I_{B → A_C}}
\def\Iow{I_{A_R → B}}
\def\Iwo{I_{B → A_R}}


Recall the definition and relevance of large allocations from \cref{prop:allocations-vs-orientations} and consider the task of finding a large allocation $μ$ for $G_n$. Imagine the vertices as agents in a parallel endeavour that proceeds in rounds and that is designed to yield information useful to construct $μ$. In each round, every vertex sends a message to each of its neighbours. Since two messages are sent between two adjacent vertices $u$ and $v$---one in each direction---it is convenient to distinguish the directed edges $(u,v)$ and $(v,u)$, the message from $u$ to $v$ being sent along $(u,v)$ and vice versa. Along $e = (u,v)$ the message is a number $I_e ∈ [0,\min(η(u),η(v))]$. We interpret this as the vertex $u$ suggesting that $μ(\{u,v\})$ be set to $I_e$. To determine $I_e$, the vertex $u$ sums up the messages it received from its \emph{other} neighbours in the previous round, obtaining a value $ξ$. If $ξ < η(u)$, then, assuming the suggestions of the neighbours of $u$ were all followed, $u$ would want $μ(e) = η(u)-ξ$ in order to fully utilise its weight $η(u)$. Taking into account the weight of $v$, $u$ sends $I_e = [η(u) - ξ]_0^{η(v)}$ where $[x]_i^j := \max(i,\min(j,x))$ is our shorthand for the “clamp function” (which we also occasionally use for one-sided clamping, leaving out the upper or lower index). Let $P$ be the operator that takes an assignment $I \colon \vec{E}(G_n) → ℕ₀$ of messages to directed edges and computes the messages $P(I) \colon \vec{E}(G_n) → ℕ₀$ of the next round.

On finite trees, iterated application of $P$ can easily be seen to converge to a unique fixed point $I^*$ of $P$, regardless of the initial assignment of messages. From $I^*$, the size of the maximum allocation can be obtained by local computations. It is plausible but non-trivial that the asymptotic behaviour of a largest allocation of $G_n$ is similarly connected to fixed points of $P$ on the random weak limit $T$ of $G_n$. The details are found in \cite{L:Belief_Propagation:2013}, but we say a few words trying to give some intuition.

Let $(u,v)$ be an edge of $T$ where $v$ is closer to the root than $u$ and let $T_u$ be the subtree of $T$ containing $v$, $u$ and all descendants of $u$. If we apply $P$ to $T$ repeatedly (starting with, say, the all-zero message assignment $I ≡ 0$), the message $I_{(u,v)}$ in later rounds will depend on ever larger parts of $T_u$ but nothing else. Assume there was a magical (measurable) function $f$ that finds, by looking at all of $T_u$, the message $f(T_u)$ that is sent along $(u,v)$ in some fixed point of $P$. In particular, if $u₁,…,u_i$ are the children of $u$, locally, the fixed point equation is \[\ f(T_u) = [η(u) - \sum_{j = 1}^i f(T_{u_j})]₀^{η(v)}. \]
Assume we have yet to reveal anything about $T_u$ and only know the types $t_u$ and $t_v$ of $u$ and $v$. Then the random variable $I_{t_u→t_v} := f(T_u)$ has a well-defined distribution. The four possible combinations of types yield random variables $\Idw,\Iow,\Iwd,\Iwo$, which must fulfil certain distributional equations.

Consider for instance $e = (u,v)$ with $u ∈ B$ and $v ∈ A_C$. On the one hand the message $f(T_u)$ is distributed like $\Iwd$. On the other hand, looking one layer deeper, $u$ has children $u₁,…,u_X$ of type $A_R$, with $X \sim \Po(ck)$ as well as one child $a' ≠ a$ of type $A_C$. The messages $f(T_{u₁}),…,f(T_{u_X})$ and $f(T_{a'})$ are independent (since the subtrees are independent) and distributed like $\Iow$ or $\Idw$, respectively, implying:
\begin{align*}
    \Iwd \stackrel{d}= \big[ℓ-\Idw - \sum_{j=1}^{X} \Iow^{(j)}\big]_0^{ℓ-1}.\tag{\hbox{$\star 3$}}
\end{align*}
where a superscript in parentheses indicates an independent copy of a random variable and “$\stackrel{d}{=}$” denotes equality in distribution.

Leconte, Lelarge, and Massoulié \cite{L:Belief_Propagation:2013} show that, remarkably, the solutions to a system of four of these equations are essentially all we require to capture the asymptotics of maximum allocations. Readers deterred by the measure theory may find Section 2 of \cite{L:A_New_Approach:2012} illuminating, which gives a high level description of the argument for a simpler case.

We now state the specialisation of the theorem that applies to our family $(G_n)_{n∈ℕ}$, which is a fairly straightforward matter with one twist: In \cite{L:Belief_Propagation:2013} allocations were restricted not only by vertex constraints (our $η\colon V(G_n) → ℕ$), but also by edge constraints giving an upper bound on $μ(e)$ for every edge $e$. We do not require them in this sense, and make all edge constraints large enough so as to never get in the way. We repurpose them \emph{for something else}, however, namely to tell apart the subtypes $A_C$ and $A_R$ within the vertex set $A$. This is because the distribution of the children of $u ∈ A$ depends on this distinction and while \cite{L:Belief_Propagation:2013} knows no subtypes of $A$ out of the box, the constraint on the edge to the parent may influence the child distribution.

For readers eager to verify the details using a copy of \cite{L:Belief_Propagation:2013}, we give the required substitutions. Let $C,R ≥ ℓ$ be two distinct constants, then the distributions $Φ_A$ and $Φ_B$ on weighted vertices with dangling weighted edges, as well as our notational substitutions are:
\begin{align*}
    Φ_A\colon &\begin{cases}
        \text{with probability $\frac{1}{1+c}$: vertex-weight $ℓ{-}1$ and $2$ edges with constraint $C$,}\\
        \text{with probability $\frac{c}{1+c}$: vertex-weight $1$ and $k$ edges with constraint $R$,}
    \end{cases}\\
    Φ_B\colon &\text{vertex-weight $ℓ$,\ \ $2$ edges with constraint $C$ and $\Po(kc)$ edges with constraint $R$},\\
    & X(C) = \Idw,\ \ X(R) = \Iow,\ \ Y(C) = \Iwd,\ \ Y(R) = \Iwo.
\end{align*}

\begin{lemma}[Special case of {\cite[Theorem 2.1]{L:Belief_Propagation:2013}}]
    \label{prop:appliedLelarge}
    \[ \lim_{n → ∞}\tfrac{M(G_{n,cn})}{n} = \inf\big\{ F(\Idw,\Iwd,\Iow,\Iwo,c) \big\}\quad \text{\AS, where} \]
    \vspace{-\baselineskip}
    \begin{gather*}
       F(\Idw,\Iwd,\Iow,\Iwo,c) = \E\bigg[
    [\Iwd^{(1)} + \Iwd^{(2)}]_0^{ℓ-1} + c · [\sum_{j=1}^k \Iwo^{(j)} ]_0^1\\
+ \big[ℓ-\sum_{i=1,2}[ℓ-\Idw^{(i)} - \sum_{j=1}^{X} \Iow^{(j)}]_0 - \sum_{i = 1}^{X} [ℓ-\Idw^{(1)}-\Idw^{(2)} - \sum_{j≠i} \Iow^{(j)}]₀\big]₀\bigg]
    \end{gather*}
    and the infimum is taken over distributions of $\Idw,\Iow,\Iwd,\Iwo$ fulfilling
    \begin{align*}
    \Idw &\stackrel{d}= ℓ-1-\Iwd,\tag{$\star 1$}\\
    \Iow &\stackrel{d}= \big[1-\sum_{j=1}^{k-1} \Iwo^{(j)}\big]_0,\tag{\hbox{$\star 2$}}\\
    \Iwd &\stackrel{d}= \big[ℓ-\Idw - \sum_{j=1}^{X} \Iow^{(j)}\big]_0^{ℓ-1},\tag{\hbox{$\star 3$}}\\
    \Iwo &\stackrel{d}= \big[ℓ-\Idw^{(1)}-\Idw^{(2)} - \sum_{j=1}^{X} \Iow^{(j)}\big]_0^1,\tag{\hbox{$\star 4$}}
    \end{align*}
    and where $X \sim \Po(kc)$ and superscripts in parentheses indicate independent copies.
\end{lemma}

To appreciate the usefulness of Lemma \ref{prop:appliedLelarge}, understanding its form is more important than understanding the significance of the individual terms.

\def\pr{\vec{ρ}}
\def\prt{\vec{ρ}_{\textrm{triv}}}

If $X$ is a random variable on a finite set $D$, then the distribution of $X$ is captured by real numbers $(\Pr[X = i])_{i ∈ D} ∈ [0,1]^{|D|}$ that sum to $1$. In this sense, the four distributions of $\Idw,\Iow,\Iwd,\Iwo$ are given by numbers $\pr ∈ [0,1]^{ℓ-1} × [0,1]² × [0,1]^{ℓ-1} × [0,1]² = [0,1]^{2ℓ+2}$. We say $\pr ∈ [0,1]^{2ℓ+2}$ is a \emph{solution} to the system $(\star)$ if the four groups of numbers belonging to the same distribution each sum to $1$ and if setting up the four random variables according to $\pr$ satisfies $(\star 1)$,$(\star 2)$,$(\star 3)$ and $(\star 4)$.

If we treat $c$ as a variable instead of as a constant, we obtain the \emph{relaxed system} $(\star c)$ where solutions are pairs $(\pr,c) ∈ [0,1]^{2ℓ+2} × (0,∞)$. The value $\prt$ that corresponds to
\[  1 = \Pr[\Idw = 0] = \Pr[\Iow = 0] = \Pr[\Iwd = ℓ-1] = \Pr[\Iwo = 1] \]
is easily checked to give rise to a solution $(\prt,c)$ of the relaxed system for any $c > 0$, we call such a solution \emph{trivial}. Evaluating $F$ for a trivial solution yields $ℓ-1+c$ so Lemma \ref{prop:appliedLelarge} implies the trivial assertion $\lim \frac{M(G_n)}{n} ≤ ℓ-1+c$ for all $c > 0$.

We now give a “nice” characterisation of the space of non-trivial solutions for $(\star c)$.

\begin{lemma}
    \label{lem:parameterise-solutions}
    For any $k,ℓ ≥ 2$, there is a bijective map $λ ↦ (\pr_λ, c_λ)$ from $(0,∞)$ to the set of non-trivial solutions for $(\star c)$.
    
    Moreover, (each component of) this map is an explicit real analytic function.
\end{lemma}





\begin{proof}
Note that $Y := \sum_{j=1}^{X} \Iow^{(j)}$ is the sum of $X$ independent indicator random variables, where $X \sim \Po(kc)$ and $\Pr[\Iow^{(j)} = 1] = q$ for some $q ∈ [0,1]$ and all $1 ≤ j ≤ X$, with $q = 0$ only occurring in trivial solutions. It is well known that such a “thinned out” Poisson distribution is again Poisson distributed and we have $Y \sim \Po(λ)$ for $λ = kcq$. Thus, each non-trivial solution to $(\star c)$ has such a parameter $λ>0$. We will now show that, conversely, $λ$ uniquely determines this solution. From $(\star 1)$ and $(\star 3)$ we obtain:
\[ \Iwd \stackrel{d}= \big[ℓ-(ℓ-1-\Iwd) - Y\big]_0^{ℓ-1} = \big[\Iwd + 1 - Y\big]_0^{ℓ-1}. \]
With $p_i := \Pr[\Iwd = i]$ for $0 ≤ i ≤ ℓ-1$ we can write this equation in matrix form as

\begin{equation}
\label{eq:matrix}
\begin{pmatrix}
    p₀\\p₁\\\vdots\\p_{ℓ-1}
\end{pmatrix}
= e^{-λ} ·
\begin{pmatrix}
    * & * & * & … & *\\
    1 & λ & λ²/2 & … & \frac{λ^{ℓ-1}}{(ℓ-1)!}\\
    & \ddots & \ddots & & \vdots\\
    & & 1 & λ &λ²/2\\
    & & & 1 & λ+1\\
\end{pmatrix}\begin{pmatrix}
    p₀\\p₁\\\vdots\\p_{ℓ-1}
\end{pmatrix}
\end{equation}
This uses $\Pr[Y = j] = e^{-λ}λ^j/j!$ for $j ∈ ℕ₀$. Each “$*$” is such that the columns of the matrix sum to $(e^λ,…,e^λ)$, which is implicit in the fact that we deal with distributions. The unique solution for a fixed $λ$ can be obtained by using the equations from bottom to top to express $p_{ℓ-2},p_{ℓ-3},…,p₀$ in terms of $p_{ℓ-1}$ and then choosing $p_{ℓ-1}$ such that the probabilities sum to $1$.
This yields a closed expression $p_j = p_j(λ)$ for $0 ≤ i ≤ ℓ-1$.


Using first $(\star 1)$, then $(\star 4)$ (with the definition of $Y$) and finally $(\star 2)$ the distributions of $\Idw$, $\Iwo$ and $\Iow$ fall into place, completing the unique solution candidate $\pr_{λ}$. The only loose end is the definition of $λ$, which gives a final equation:
If $q(λ) = \Pr[\Iow = 1]$ is the value we computed after choosing $λ$, we need $λ = kc·q(λ)$, which uniquely determines a value $c_{λ} = \frac{λ}{k·q(λ)}$ (it is easy to check that $λ > 0$ guarantees $q(λ) >0$). Thus $(\pr_{λ},c_{λ})$ is the unique solution with parameter $λ$. Retracing our steps it is easy to verify that we only composed real analytic functions.
\end{proof}
With the parametrisation of the solutions of $(\star c)$, we can, with a slight stretch of notation, rewrite Lemma \ref{prop:appliedLelarge}. For any $c > 0$ we have
\begin{equation}
    \lim_{n → ∞}\tfrac{M(G_{n,cn})}{n} = \inf(\{ F(\pr_{λ},c_{λ}) \mid λ ∈ (0,∞), c_λ = c\} ∪ \{ℓ-1+c\})\quad \text{\AS}. \label{eq:appliedLelarge-improved}
\end{equation}
We now define the value $γ_{k,ℓ}$ and by proving \cref{prop:threshold-up-to-on} demonstrate its significance.
\begin{equation}
    γ_{k,ℓ} := \inf_{λ > 0}\{c_{λ} \mid F(\pr_{λ},c_{λ}) < ℓ-1+c_{λ} \}.\label{eq-def:threshold}
\end{equation}


\begin{proof}[Proof of \cref{prop:threshold-up-to-on}]
    \begin{description}
        \item[\textrm{Case} $\bm{c < γ_{k,ℓ}}$.] By definition of $γ_{k,ℓ}$ there is no parameter $λ$ with $c_{λ} = c$ and $F(\pr_{λ},c_{λ}) < ℓ-1+c$. Thus, \cref{eq:appliedLelarge-improved} implies $\lim_{n → ∞}\frac 1n {M(G_{n,cn})} = ℓ-1+c$ \AS.
        \item[\textrm{Case} $\bm{c > γ_{k,ℓ}}$.]
            By definition of $γ_{k,ℓ}$, for $c = γ_{k,ℓ} + ε$ there is some $λ$ with $c_{λ} ∈ [γ_{k,ℓ},c)$ and $F(\pr_{λ},c_{λ}) ≤ ℓ-1+c_{λ}-ε'$ for some $ε' > 0$. This implies $\lim_{n→∞}\frac 1n {M(G_{n,c_{λ}n})} ≤ ℓ-1+c_{λ}-ε'$ \AS. 
    
            Since $G_{n,cn}$ can be obtained from $G_{n,c_λn}$ by adding $(c-c_{λ})n$ vertices of weight $1$ with random connections, and this can increase the size of a maximum allocation by at most $(c-c_{λ})n$, we also have $\lim_{n→∞}\frac 1n {M(G_{n,cn})} ≤ ℓ-1+c-ε'$ \AS.\qedhere
    \end{description}
\end{proof}


\section{Closing the Gap – Proof of the Main Theorem}
\label{sec:mainproof}

The key ingredient still missing to prove the main theorem is Lemma \ref{lem:no-small-hall-wittnesses}, stated on page \pageref{lem:no-small-hall-wittnesses}.
\begin{proof}[Proof of Lemma \ref{lem:no-small-hall-wittnesses}]
    Call a set $X ⊂ ℤ_n = V(\hut W_n)$ a \emph{bad set} if it induces (hyper-)edges of total weight $ℓ|X|$ or more. We now consider all possible sizes $t$ of $X$ and each possible number $αt$ ($0 < α ≤ 1$) of contiguous segments of $X$ separately, using the first moment method to bound the probability that a bad set $X$ with such parameters exists, later summing over all $t$ and $α$.
    
    For now, let $t$ and $α$ be fixed and write $X$ as the union of non-empty, non-touching\footnote{Two intervals touch if their union is an interval, i.e.\ if there is no gap in between them.} intervals $X = X₁ ∪ X₂ ∪ … ∪ X_{αt}$ arranged on the cycle $ℤ_n$ in canonical ordering and with $X₁$ being the interval containing $\min X$. We write the complement $ℤ_n - X = Y₁ ∪ Y₂ ∪ … ∪ Y_{αt}$ in a similar way. It is almost possible to reconstruct $X$ from the sets $\{ x₁,…,x_{αt}\}$ and $\{y₁,…,y_{αt}\}$ where $x_i := |X₁ ∪ … ∪ X_i|$ and $y_i := |Y₁ ∪ … ∪ Y_i|$, we just do not know where $X₁$ starts. To fix this, we exploit that $x_{αt}$ and $y_{αt}$ are always $t$ and $n-t$, respectively, and do not really encode information. In the case $0 ∈ X$, we set $x'_{αt} := \max X₁ ∈ [0,…,x₁)$ and $y'_{αt} := y_{αt}$; if $0 ∉ X$ we set $y'_{αt} := \max Y₁ ∈ [0,…,y₁)$ and $x'_{αt} := x_{αt}$. The sets $\{x₁,…,x_{αt-1},x'_{αt}\} ⊆ \{0,…,t\}$ and $\{y₁,…,y_{αt-1},y'_{αt}\} ⊆ \{0,…,n-t\}$ now uniquely identify $X$, meaning there are at most $\binom{t+1}{αt}\binom{n-t+1}{αt}$ choices for $X$. No matter the choice, $X$ induces helper edges of weight precisely $(ℓ-1)·(t-αt)$, since for each $x ∈ X$ the edge $(x,x+1)$ is induced, except if $x$ is the right endpoint of one of the $αt$ intervals. In order for $X$ to induce a total weight of $ℓt$ or more another
    \[ ℓt - (ℓ-1)(t-αt) = t + (ℓ-1)αt ≥ t + αt \]
    ordinary edges (of weight $1$) need to be induced. There are $\binom{cn}{t+αt}$ ways to choose such a set of edges and each edge has all endpoints in $X$ with probability $(\frac{t}{n})^k ≤ (\frac{t}{n})²$. 
    
    Together we obtain the following upper bound on the probability that a bad set of size $t$ with $αt$ contiguous regions exists:
    \begin{align*}
        &\underbrace{\binom{t+1}{αt}\binom{n-t+1}{αt}}_{\text{choices for $X$}}\underbrace{\binom{cn}{t+αt}}_{\text{choices for edges}}\left(\frac{t}{n}\right)^{2(t+αt)}
        ≤2^{t+1}\binom{n}{αt}\binom{n}{t+αt}\left(\frac{t}{n}\right)^{2(t+αt)}\\
        &≤2^{t+1}\left(\frac{ne}{αt}\right)^{αt}\left(\frac{ne}{t+αt}\right)^{t+αt}\left(\frac{t}{n}\right)^{2(t+αt)}
        ≤2\left(2 \left(\frac{ne}{αt}\right)^{α} \left(\frac{ne}{t+αt}\right)^{1+α} \left(\frac{t}{n}\right)^{2(1+α)} \right)^t\\
        &≤2\bigg(\underbrace{2\left(\frac{e}{α}\right)^{α} \left(\frac{e}{1+α}\right)^{1+α}}_{≤ C \text{ for some $C = O(1)$}} \left(\frac{n}{t}\right)^{α} \left(\frac{n}{t}\right)^{1+α} \left(\frac{t}{n}\right)^{2(1+α)} \bigg)^t
        ≤2\left(C \left(\frac{t}{n}\right) \right)^t.
    \end{align*}
    We used that $α^α$ is bounded (has limit 1) for $α → 0$. The resulting term is $o(1)$ for $t = 1,2,3$ (and only $1,2$ or $3$ choices for $α$ are possible). For $4 ≤ t ≤ \sqrt{n}$, it is  
    \[ 2\left(C \left(\frac{t}{n}\right) \right)^t ≤ 2\left(\frac{C}{\sqrt{n}}\right)^t ≤  2\left(\frac{C}{\sqrt{n}}\right)^4 = O(n^{-2}).\]
    Summing over all combinations of $O(\sqrt{n}·\sqrt{n})$ choices for $α$ and $t$,  we get a sum of $O(n^{-1})$.
    For $\sqrt{n} ≤ t ≤ δn$ with $δ = \frac{1}{2C}$ we get
    \[ 2\left(C \left(\frac{t}{n}\right) \right)^t ≤ 2\left(\frac 12\right)^t ≤ 2·2^{-\sqrt{n}},\]
    which is clearly $o(1)$ even if we sum over all $O(n²)$ combinations for choosing $t$ and $α$.
\end{proof}

\begin{proof}[Proof of the Main Theorem]
    It suffices to show that $γ_{k,ℓ}$ is the threshold for the event $\{M(G_n) = n(ℓ-1+c)\}$ since by \cref{prop:eliminate-cycles,prop:allocations-vs-orientations} this event coincides with the events that $\hat{W}_n$ and $W_n$ are orientable.
    
    \begin{description}
        •[$\bm{c > γ_{k,ℓ}}$.] If $c = γ_{k,ℓ} + ε$ for $ε > 0$, then by \cref{prop:threshold-up-to-on} we have $\lim_{n→∞}\frac{M(G_{n,cn})}{n} = ℓ-1+c-ε'$ \AS for some $ε' >0$. This clearly implies that $M(G_{n,cn}) < n(ℓ-1+c)$ \whp.
        •[$\bm{c < γ_{k,ℓ}}$.] Let $c = γ_{k,ℓ} - ε$ for some $ε > 0$ and define $c' := γ_{k,ℓ}-\frac{ε}2$. We generate $G_{n,cn}$ from $G_{n,c'n}$ by removing $\frac{ε}2n$ vertices from $A_R$. The idea is to derive orientability of $G_{n,cn}$ from Lemma \ref{lem:no-small-hall-wittnesses} and “almost-orientability” of $G_{n,c'n}$.
        
        More precisely, let $G^{(0)} := G_{n,c'n} = (A = A_C ∪ A_R, B, \text{“$\ni$”})$ and let $G^{(i+1)}$ be obtained from $G^{(i)}$ by removing a vertex from $A_R ∩ V(G^{(i)})$ uniformly at random, $0 ≤ i < \frac{ε}2n$.
        By \cref{prop:threshold-up-to-on} we have $\lim_{n→∞}\frac{M(G_{n,c'n})}{n} = ℓ-1+c'$ \AS, which implies $M(G_{n,c'n}) = n(ℓ-1+c')-o(n) = η(A) - o(n)$ \whp. For any subgraph $G$ of $G^{(0}$, define $\gap(G) := η(A ∩ V(G)) - M(G)$ which measures how far $G$ is away from being orientable.
        This ensures $\gap(G^{(0)}) = o(n)$ \whp as well as $\gap(G^{(i+1)}) ∈ \{\gap(G^{(i)}), \gap(G^{(i)}) - 1\}$. We say a vertex $a ∈ V(G^{(i)}) ∩ A_R$ is \emph{good} for $G^{(i)}$ if $\gap(G^{(i)}-a) = \gap(G^{(i)})-1$.
        
        Assume $\gap(G^{(i)}) > 0$. We now show that $Θ(n)$ vertices are good for $G^{(i)}$ \whp. Let $a ∈ A$ be one vertex of $G^{(i)}$ that is not saturated in a maximum allocation $μ$ of $G^{(i)}$. Let $X ⊆ B$ and $Y ⊆ A$ be the vertices from $A$ and $B$ reachable from $a$ via an \emph{alternating path}, i.e.\ a path of the form $(a = a₁,b₁,a₂,b₂,…)$ such that $μ(b_j,a_{j+1}) > 0$ for all $j$.
        
        It is easy to check that all vertices from $X$ are saturated in $μ$ (otherwise $μ$ could be increased), $Y$ exceeds $X$ in total weight and every vertex from $Y ∩ A_R$ is good for $G^{(i)}$. Moreover, when viewed as a subset of $V(\hat{W}_{n,c'n})$, $X$ induces at least $Y$. Discounting the low probability event that Lemma \ref{lem:no-small-hall-wittnesses} does not apply to $\hat{W}_{n,c'n}$, we conclude $|X| > δn$, and taking into account $|Y ∩ A_C| ≤ |X|$ (clear from definition of $A_C$) together with $(ℓ-1)|Y ∩ A_C|+|Y ∩ A_R| = η(Y) > η(X) = ℓ|X|$ we obtain $|Y ∩ A_R| > δn$.

        This means that \whp on the way from $G^{(0)}$ to $G^{(\frac{ε}2n)}$, we start with a gap of $o(n)$ and have up to $\frac{ε}2n = Ω(n)$ chances to reduce a non-zero gap by $1$, namely by choosing a good vertex for removal, and the probability is at least $δ = Ω(1)$ every time. Now simple Chernoff bounds imply that the gap vanishes \whp meaning $G^{(\frac{ε}2n)}$ is orientable $\whp$. Since the distributions of $G_{n,cn}$ and $G^{(\frac{ε}2n)}$ coincide, we are done.\qedhere
    %
    \end{description}
\end{proof}


\section{Numerical approximations of the Thresholds.} 
\label{sec:approximations}
Rewriting the definition of $γ_{k,ℓ}$ in \cref{eq-def:threshold} in the form promised in the introduction we get
\[ γ_{k,ℓ} = \inf_{λ > 0}\{f_{k,ℓ}(λ) \mid g_{k,ℓ}(λ) < 0 \}\]
where $f_{k,ℓ}(λ) = c_λ$ and $g_{k,ℓ}(λ) = F(\pr_{λ},c_{λ}) - ℓ + 1 - c_{λ}$. We give the plots of $f_{2,2}$ and $g_{2,2}$ in \cref{fig:plots-of-f-and-g}.

\begin{figure}[hbtp]
    \includegraphics[scale=0.7]{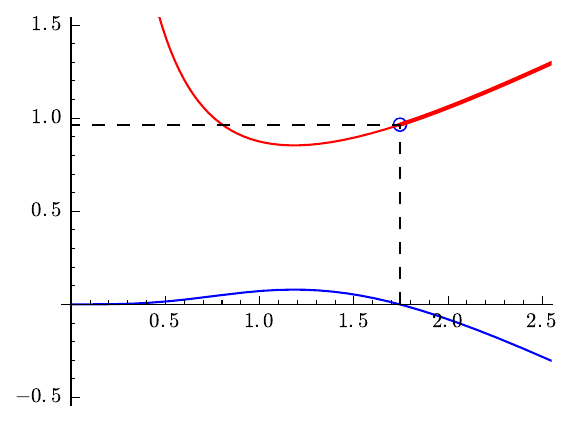}
    \Description{The plot of $g_{2,2}(λ)$ starts at $0$ for $λ = 0$, slowly increases to a maximum, then dips below zero at around $λ = 1.74$. The plot of $f_{2,2}(λ)$ looks convex with a minimum where $g_{2,2}(λ)$ has its maximum. The value at $λ = 1.74$ is highlighted and roughly $f_{2,2}(λ) ≈ 0.965$.}
    \caption{The functions \textcolor{red}{$f_{2,2}(λ)$} and \textcolor{blue}{$g_{2,2}(λ)$}, with $γ_{2,2} = \inf_{λ > 0}\{f_{2,2}(λ) \mid g_{2,2}(λ) < 0 \}$.}
    \label{fig:plots-of-f-and-g}
\end{figure}

It looks as though $\{λ \mid g_{2,2}(λ) < 0\}$ is an interval $(λ^*\,{≈}\,1.74, ∞)$ on which $f_{2,2}$ is monotonically increasing, meaning $γ_{2,2} = f_{2,2}(λ^*) ≈ 0.965$. Here is a semi-rigorous argument that properly done plots cannot be misleading: Regardless of $k,ℓ≥2$ it is fairly easy to see that $c_λ = Ω(λ^{-(k-1)ℓ+1})$ for $λ→0$ and $c_λ = Ω(λ)$ for $λ → ∞$. In particular, for fixed $k,ℓ$ it is easy to obtain bounds $0 < λ₀ < λ₁ < ∞$ such that for $λ ∈ (0,λ₀) ∪ (λ₁,∞)$ we can guarantee $f_{k,ℓ}(λ) > 1$. Because $γ_{k,ℓ} < 1$, only the interval $[λ₀,λ₁]$ can be relevant. Being real analytic, the functions $f_{k,ℓ}$ and $g_{k,ℓ}$ have bounded first and second derivatives on $[λ₀,λ₁]$ which vindicates plots of sufficient resolution: There cannot be unexpected zeroes in between sampled positions and what looks strictly monotonic in the plot actually is. So starting with a golden ratio search close to the apparent root of $g_{2,2}$ we are guaranteed to find $λ^*$ and $γ_{2,2} = f_{2,2}(λ^*)$. This can be made formal.

While handling it this way saves us the trouble of having to deal with unwieldy functions, our lack of analytical insight means we have to consider each pair $(k,ℓ)$ separately to make sure that $f_{k,ℓ}$ and $g_{k,ℓ}$ do not exhibit a qualitatively different behaviour. We did this for $ℓ ≤ 4$ and $k ≤ 7$, i.e.\ for the values we provided in table \cref{table:values}.


\section{Speed of Convergence and Practical Table Sizes}
\label{sec:exp-speed-of-convergence}

\begin{figure}[htbp]
    \includegraphics[width=0.33\textwidth]{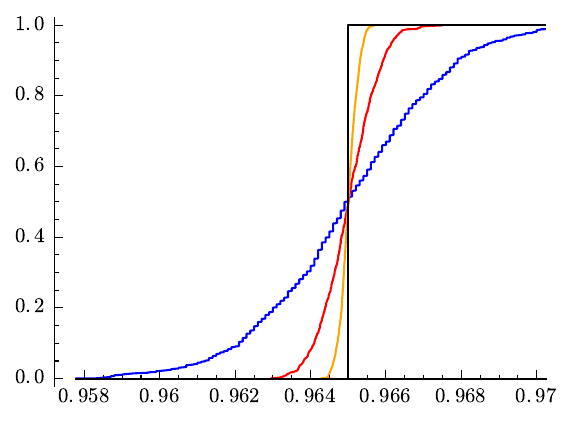}\includegraphics[width=0.33\textwidth]{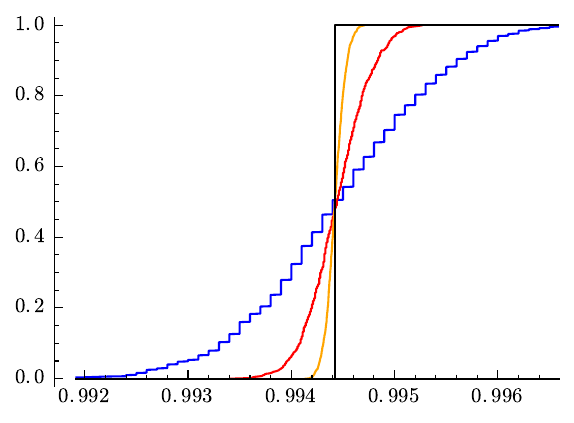}\includegraphics[width=0.33\textwidth]{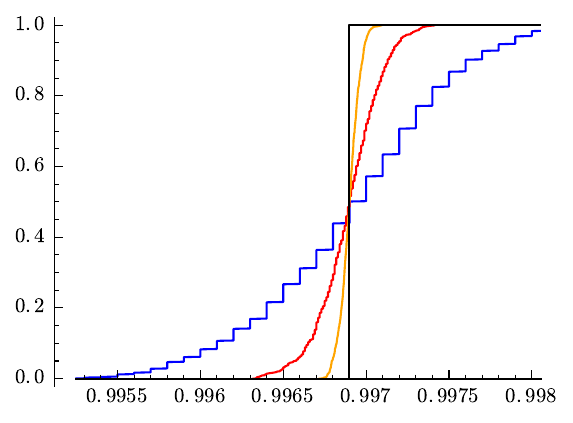}
    \Description{Three plots of three colour-coded functions each. In each case, the functions get steeper but seem assume the value $1/2$ at roughly the same point. The function for $n = 10^4$ is still visibly discrete, the others look smooth.}
    \caption{Approximate probabilities for $\hat{W}_{n,cn}^{k,ℓ}$ to not be orientable depending on $c$. From left to right, the plots correspond to $(k,ℓ) ∈ \{(2,2), (2,3), (3,2)\}$. In each plot the three curves correspond to $n ∈ \{ \textcolor{blue}{10^4}, \textcolor{red}{10^5}, \textcolor{orange}{10^6} \}$, with curves for larger $n$ visibly getting closer to the step function that we know is the limit for $n → ∞$. Details are given in the text.}
    \label{fig:speedOfConvergence}
\end{figure}

\noindent It is natural to wonder to what degree the asymptotic results of this paper predict the behaviour of hash tables for realistic (large-ish but fixed) values of $n$, say a hash table of size $n = 10^5$. Formally, for $k,ℓ ≥ 2$ and $t ∈ [0,1]$ we define the functions
\[ p_n^{k,ℓ}(c) := \Pr[\hat{W}_{n,cn}^{k,ℓ} \text{ is \emph{not} orientable}] \quad \text { and } \quad \mathrm{step}_{t}(c) = \begin{cases}
    0 & \text{ if $c < t$}\\
    \frac 12 & \text{ if $c = t$}\\
    1 & \text{ if $c > t$}
\end{cases}\]
This paper shows that $p_n^{k,ℓ} \colon [0,1] → [0,1]$ converges point-wise to $\mathrm{step}_{γ_{k,ℓ}}$, except, possibly, at the point of the threshold $c = γ_{k,ℓ}$ itself. To give an idea of the speed of this convergence we plotted approximations of $p_n^{k,ℓ}$ for $n ∈ \{10^4,10^5,10^6\}$ and $(k,ℓ) ∈ \{(2,2),(2,3),(3,2)\}$ in \cref{fig:speedOfConvergence}.

To obtain the approximation of $p_n^{k,ℓ}$, we carried out $1000$ trials. In each trial, a graph of the form $\hat{W}_{n,m}^{k,ℓ}$ was generated by adding random\footnote{Using Pseudo-Random numbers produced by the \textsc{mt19937} implementation of the \textsc{c++} standard library.} edges one by one until the graph was not orientable any more. The number of edges $m$ where orientability breaks down corresponds to the load $c = \frac{m}{n}$. As an estimate for $p_n^{k,ℓ}(c)$ we take the fraction of the $1000$ trials where orientability broke at a value less than $c$.

The plots suggest that, at the threshold, $p_n^{k,ℓ}$ assumes a value of $\frac 12$ and has a slope proportional to $\sqrt{n}$.

\section{Linear Time Construction of an Orientation?}
\label{sec:khosla-lsa}

Constructing a placement of objects in a hash table is not the focus of this paper, but we provide experimental support for the following conjecture:

\begin{conjecture}
    \label{conj:khosla-works}
    Let $k ≥ 2$ and $ℓ ≥ 2$ and $c = γ_{k,ℓ} - ε$ for some $ε > 0$. An orientation of $\hat{W}^{k,ℓ}_{n,cn}$ or, equivalently, a placement of $cn$ objects into $n$ table cells for $k$-ary cuckoo hashing with windows of size $ℓ$ can be constructed in time $O(C(ε)n)$ \whp using the adapted LSA Algorithm explained below. Here, $C(ε)$ does not depend on $n$.
\end{conjecture}


Algorithms to orient certain random graphs have been known for a while, for instance the \emph{selfless algorithm} analysed by Cain, Sanders and Wormald \cite{CSW:The_Random:2007} or an algorithm by Fernholz and Ramachandran \cite{FR:The_k-orientability:2007} involving so-called \emph{excess degree reduction}. While a generalisation of the selfless algorithm has been suggested by Dietzfelbinger et al.~\cite{DGMMPR:Tight:2010}, the algorithm that seemed easiest to adapt to our particular hypergraph setting is the Local Search Allocation (LSA) algorithm by Khosla \cite{Khosla:Balls-Into-Bins:2013} (with improved analysis in \cite{Khosla:Balls-In-Bins-Extended:2019}).

\begin{figure}[hbtp]
    \includegraphics[scale=0.65]{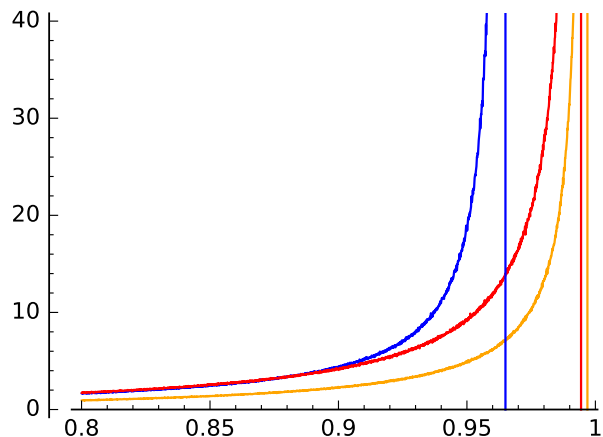}
    \Description{Three curves plotted on the interval $[0.8,1]$. They are fairly smooth but some noise is visible. They start with values around $2$ or so at $0.8$. Each curve bends upwards and leaves the image (assumes a value exceeding $40$) roughly $0.01$ or so left of a correspondingly coloured vertical line.}
    \caption{Average number of balls touched (one plus the number of evictions) to insert a new ball at a certain load using Khosla's LSA Algorithm (each point in the plot being the average of $10000$ insertions). The blue, red and orange curves correspond to $k$-ary cuckoo hashing with windows of size $ℓ$ for $(k,ℓ) ∈ \{\textcolor{blue}{(2,2)},\textcolor{red}{(2,3)},\textcolor{orange}{(3,2)}\}$, respectively, and a table of size $n = 10^7$. The corresponding thresholds are shown as vertical lines. Analogous plots for $n = 10^6$ are visually indistinguishable, suggesting the plots are stable as $n$ varies.}
    \label{fig:placement-times}
\end{figure}

\paragraph{The Algorithm.} Following Khosla's terminology, we describe the task of orienting $\hat{W}^{k,ℓ}_{n,cn}$ as a problem of placing balls into bins. There are $n$ bins of capacity $ℓ$ arranged in a circle and for each pair of adjacent bins there are $(ℓ-1)$ helper balls that may be placed into one of those two bins. Moreover there are ordinary balls, each of which has $k$ random bins it may be placed into.
 
 We start with all helper balls placed in their “left” option, in particular all bins have room for just one more ball. Now the ordinary balls are placed one by one. We maintain a label for each bin. All labels are natural numbers, initially zero. We place a ball simply by putting it in the admissible bin with least label. If placing a ball results in an overloaded bin $b$, one ball must be evicted from $b$ and inserted again. The ball to be evicted is chosen to have, among all balls in $b$, an alternative bin of least label.
 
 Whenever the content of a bin $b$ changed (after insertion or insertion + eviction), its label is updated. The new label is one more than the least label of a bin that is the alternative bin of a ball currently placed in $b$.
 
\paragraph{Analysis.} Labels can be thought of as lower bounds on the distance of a bin to the closest non-full bin in the directed graph where bins are vertices and an edge from $b₁$ to $b₂$ indicates that a ball in $b₁$ has $b₂$ as an alternative bin. It is fairly easy to see that this algorithm finds a placement in quadratic time whenever a placement exists. To show that running time is linear \whp, it suffices to show that the running time is linear in the sum of all labels in the end and that the sum of the aforementioned distances is linear \whp. We don't attempt a proof here, although we expect it to be possible with Khosla's techniques.

\paragraph{Experiments.} The results from \cref{fig:placement-times} suggest that the expected number of evictions per insertion is bounded by a constant as long as the load $c$ is bounded away from the threshold. Eviction counts sharply increase close to the threshold.


\section{Conclusion and Outlook}

We established a method to determine load thresholds $γ_{k,ℓ}$ for $k$-ary cuckoo hashing with (unaligned) windows of size $ℓ$. In particular, we resolved the cases with $k = 2$ left open in \cite{DW07:Balanced:2007,LP:3.5-Way:2009}, confirming corresponding experimental results by rigorous analysis. 

The following four questions may be worthwhile starting points for further research.

\paragraph{Is there more in this method?} It is conceivable that there is an insightful simplification of Lemma \ref{prop:appliedLelarge} that yields a less unwieldy characterisation of $γ_{k,ℓ}$. We also suspect that the threshold for the appearance of the $(ℓ+1)$-\emph{core} of $\hut W_n$ can be identified with some additional work (for cores see e.g. \cite{Molloy05:Cores-in-random-hypergraphs,Luczak:A-simple-solution}). This threshold is of interest because it is the point where the simple peeling algorithm to compute an orientation of $\hut W_n$ breaks down.


\paragraph{Can we prove efficient insertion?} Given our experiments concerning the performance of Khosla's LSA algorithm for inserting elements in our hashing scheme, it seems likely that its running time is linear, see \cref{conj:khosla-works} in \cref{sec:khosla-lsa}. But one could also consider approaches that do not insert elements one by one but build a hash table of load $c = γ_{k,ℓ} - ε$ given all elements at once. Something in the spirit of the selfless algorithm \cite{CSW:The_Random:2007} or excess degree reduction \cite{DGMMPR:Tight:2010} may offer linear running time with no performance degradation as $ε$ gets smaller, at least for $k = 2$.



\paragraph{How good is it in practice?} This paper does not address the competitiveness of our hashing scheme in realistic practical settings. The fact that windows give higher thresholds than (aligned) blocks for the \emph{same} parameter $ℓ$ may just mean that the “best” $ℓ$ for a particular use case is lower, not precluding the possibility that the associated performance benefit is outweighed by other effects. \cite{DW07:Balanced:2007} provide a few experiments in their appendix suggesting slight advantages for windows in the case of unsuccessful searches and slight disadvantages for successful searches and insert operations, in one very particular setup with $k = 2$. Further research could take into account precise knowledge of cache effects on modern machines, possibly using a mixed approach, respecting alignment only insofar as it is favoured by the caches. Ideas from Porat and Shalem \cite{PS:A_Cuckoo_Hashing:2012} could prove beneficial in this regard.

\paragraph{What about other geometries?} We analysed linear hash tables where objects are assigned random intervals. One could also consider a square hash table $(ℤ_{\sqrt{n}})^2$ where objects are assigned random squares of size $ℓ × ℓ$ (with no alignment requirement). We suspect that understanding the thresholds in such cases would require completely new techniques.





\begin{acks}
I am indebted to my advisor Martin Dietzfelbinger for drawing my attention to this problem as well as for providing a constant stream of useful literature recommendations. When discussing a preliminary version of this work at the Dagstuhl Seminar 17181 on Theory and Applications of Hashing, Michael Mitzenmacher and Konstantinos Panagiotou provided useful comments. The full version of this paper also profited from helpful reviewers who, among other things, pointed out the need for the discussion that is now \cref{sec:approximations,sec:exp-speed-of-convergence,sec:khosla-lsa} as well as more details in \cref{sec:local-weak-convergence}.
\end{acks}

    \bibliography{bibliographie}

\end{document}